\documentclass[smallextended]{svjour_blank}
\smartqed  

\usepackage{amsfonts}
\usepackage{amsmath}
\usepackage{amssymb}
\usepackage{algorithm}
\usepackage[noend]{algpseudocode}
\usepackage{xcolor}
\usepackage{graphicx}

\newcommand{\dw}{\nobreak{D-Wave} Two}

\newcommand{\fc}{\mathcal{C}}

\newcommand{\be}{\begin{equation}}
  \newcommand{\ee}{\end{equation}}
\newcommand{\bea}{\begin{eqnarray}}
  \newcommand{\eea}{\end{eqnarray}}

\newcommand{\br}[1]{\left\{{#1}\right\}}

\newcommand{\abs}[1]{\left|{#1}\right|}
\newcommand{\norm}[1]{\left\|{#1}\right\|}

\newcommand{\cX}{\mathcal{X}}
\newcommand{\cB}{\mathcal{B}}

\renewcommand{\epsilon}{\varepsilon}

\begin{document}

\title{Fast clique minor generation in Chimera qubit connectivity graphs\thanks{This research was partially supported by the Mitacs Accelerate program.}}

\author{Kelly Boothby         \and
  Andrew D.~King        \and
  Aidan Roy
}

\institute{K. Boothby (\email{boothby@dwavesys.com}) 
  \and
  A. King (corresponding author, \email{aking@dwavesys.com})
  \and  A. Roy (\email{aroy@dwavesys.com})
  \at D-Wave Systems Inc., 3033 Beta Avenue, Burnaby, BC, Canada, V5G-4M9 
}

\date{\today}

\maketitle

\begin{abstract}
The current generation of D-Wave quantum annealing processor is designed to minimize the energy of an Ising spin configuration whose pairwise interactions lie on the edges of a {\em Chimera} graph $\fc_{M,N,L}$.  In order to solve an Ising spin problem with arbitrary pairwise interaction structure, the corresponding graph must be minor-embedded into a Chimera graph.  We define a combinatorial class of {\em native clique minors} in Chimera graphs with vertex images of uniform, near minimal size, and provide a polynomial-time algorithm that finds a maximum native clique minor in a given induced subgraph of a Chimera graph.  These minors allow improvement over recent work and have immediate practical applications in the field of quantum annealing.
\end{abstract}

\section{Introduction and motivation}

D-Wave quantum annealing processors are designed to sample low-energy spin configurations in the Ising model using open-system quantum annealing \cite{Albash2015,Boixo2014,Dickson2013,Johnson2011}.  Input to the processor consists of an {\em Ising Hamiltonian} $(h,J)$, where $h \in \mathbb R^n$ is a vector of {\em local fields} and $J\in \mathbb R^{n\times n}$ is a matrix of {\em couplings}, which we assume here to be symmetric.  The {\em energy} of a spin configuration $s\in \{-1,1 \}^n$ is given as 
\begin{equation}
  E(s) = E(h,J,s) = s^TJ s + s^Th.
\end{equation}
The output of an {\em anneal} (i.e.\ a run) of the processor is a low-energy state $s$, which consists of an Ising spin (either $-1$ or $1$) for each {\em qubit}.

Nonzero terms of $J$ are physically realized using programmable {\em couplers}.  These couplers only exist between certain pairs of physically proximate qubits.  The input $(h,J)$ is therefore restricted such that if $J_{i,j}\neq 0$, there must be a coupler between qubit $i$ and qubit $j$.  We rephrase this in graph-theoretic terms: The {\em connectivity graph} of $(h,J)$, which we denote by $G_J$, is the undirected graph on $n$ vertices whose adjacency matrix has the same nonzero entries as $J$.  Likewise, each processor has a {\em hardware graph} $G$ representing the available qubits and couplers in the processor.  For $(h,J)$ to be input directly to the processor, $G_J$ must be a subgraph of $G$.  If this is not the case, we can indirectly input $(h,J)$ to the hardware by embedding $G_J$ as a graph minor of $G$.  Implementing graph minors in the Ising model involves putting a strong ferromagnetic coupling $J_{i,j}\ll 0$ between any two adjacent vertices $i,j$ of $G$ in the same vertex image.  This coupling compels multiple qubits to take the same spin, thus acting like a single logical qubit.  The method is studied in greater detail elsewhere \cite{Choi2008,Venturelli2014,Perdomo-Ortiz2015,King2014,Cai2014}.

In this paper we consider the problem of finding large clique minors in the hardware graph.  This is sufficient for minor-embedding any problem of appropriate size in a given hardware graph, and allows the study of random, fully-connected spin glass problems, as in recent work \cite{Venturelli2014}.  Klymko et al.\ first provided a polynomial-time algorithm for generating large clique minors in subgraphs of hardware graphs \cite{Klymko2014}.  In Section~\ref{sec:previouswork} we provide evidence that our algorithm uses fewer physical qubits and allows the embedding of larger minors.

\subsection{The Chimera graph and triangle embeddings}

D-Wave processors currently operate using a {\em Chimera} hardware graph $\fc_{M,N,L}$, which is an $M\times N$ grid of $K_{L,L}$ complete bipartite graphs (unit cells). \dw~ processors use a 512-qubit $\fc_{8,8,4}$ hardware graph, and the most recent processors use a 1152-qubit $\fc_{12,12,4}$ graph.

\begin{figure}
  \begin{center}
    \includegraphics{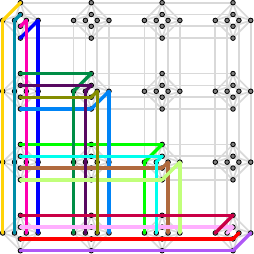}
  \end{center}
  \caption{(Color online) ``Triangle'' clique embedding in $\fc_{4,4,4}$, which motivated the design of the Chimera graph.  Every chain in this embedding has 5 qubits.}\label{fig:tridia}
\end{figure}

The Chimera graph was chosen, in part, because it contains a particularly nice clique minor \cite{Choi2011} (see Figure \ref{fig:tridia}).  This {\em triangle embedding} is uniform in the sense that each vertex image (or {\em chain}) has the same number of vertices, and is near-optimal in the sense that it gives a $K_{LM}$ minor in $\fc_{M,M,L}$, whereas $\fc_{M,M,L}$ has treewidth $LM$ and therefore contains no $K_{L M +2}$ minor\footnote{Taking the triangle embedding and making an image of all the unused qubits gives a $K_{LM +1}$ minor.}. A degree argument also shows that any uniform $K_{LM}$ minor requires chains of size at least $M$, while the triangle embedding has chains of size $M+1$.

In practice, a given processor will have a number of inoperable qubits.  If there are $t$ inoperable qubits, up to $t$ chains in a particular clique embedding can be rendered useless.  These inoperable qubits force us to find a clique minor in an induced subgraph of $\fc_{M,N,L}$.  In the face of this, note that there are at least four triangle embeddings, which we can find by simply rotating the embedding shown in Figure \ref{fig:tridia}.  So a first attempt at minimizing the impact of inoperable qubits is to choose the triangle embedding for which the greatest number of chains survive.

Triangle embeddings can be generalized further.  A triangle embedding consists of overlapping ell-shaped (L-shaped) ``bundles'' of chains, and each chain in a bundle joins a horizontal ``wire'' with a vertical ``wire'' via a matching at the corner of the bundle.  First, the structure of overlapping ell-shaped bundles can be generalized from the triangle (we can avoid all four corner unit cells, for example).  Second, the corner matchings can be chosen arbitrarily to minimize the impact of inoperable qubits (if there are three intact vertical wires and three intact horizontal wires, we can ensure that they are matched together to make three intact chains).  These generalizations result in exponential expansion of the number of clique embeddings available, but we can optimize over them in polynomial time using a dynamic programming approach.  Defining and efficiently optimizing over these clique embeddings are the main goals of this work.

In the next section we formalize the definition of {\em native clique embeddings} that generalize triangle embeddings, and give a combinatorial characterization of the same.  In Section \ref{sec:3} we give a dynamic programming technique that, given an induced subgraph of $\fc_{M,N,L}$, finds a maximum-sized native clique embedding in polynomial time.  One desirable feature of native clique embeddings is uniform chain length, which results in uniform, predictable chain dynamics throughout the anneal \cite{Venturelli2014}.  Clique embeddings found through heuristic methods such as the algorithm described by Cai et al.\ \cite{Cai2014} generally lack this property.  In Section \ref{sec:previouswork} we compare the results of our approach to those of the somewhat similar approach by Klymko et al.\ \cite{Klymko2014}.  The shorter chains and larger cliques generated by our approach lead to improved tunneling dynamics and error insensitivity \cite{Dziarmaga2005,Venturelli2014,Young2013}.

\section{Native Clique Embeddings}
\begin{figure}
  \begin{center}
    \includegraphics{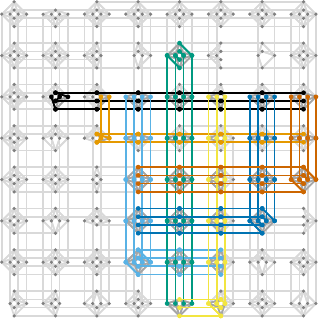}
    \includegraphics{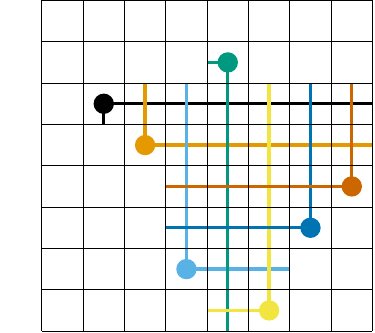}
  \end{center}
  \caption{(Color online) (\textit{l}) A maximum native clique embedding of $K_{24}$ in an induced subgraph of $\fc_{8,8,4}$ with 26 randomly selected vertices deleted. (\textit{r}) The corresponding block clique embedding, with dots indicating corners of the ell blocks. }\label{fig:nativeembedding}
\end{figure}

We now formally define the structures required to construct and analyze native clique embeddings.

Recall that Chimera $\fc_{M,N,L}$ is an $M\times N$ grid of $K_{L,L}$ unit cells.  Specifically, $\fc_{M,N,L}$ has vertices $V = \{1,\cdots,M\} \times \{1,\cdots,N\} \times \{0,1\} \times \{1,\cdots,L\}$ and edges:

\begin{tabular}{ll}
$(x,y,0,k) \sim (x+1,y,0,k)$ & (horizontal inter-cell couplings), \\
$(x,y,1,k) \sim (x,y+1,1,k)$& (vertical inter-cell couplings), and \\
$(x,y,0,k_1) \sim (x,y,1,k_2)$ & (intra-cell couplings).
\end{tabular}


We construct native clique embeddings using {\em wires}. For $t\geq 1$, a {\em horizontal wire} of length $t$ is a contiguous set of vertices $\{(x+i,y,0,k): i\in[0,t-1]\}$, whose induced subgraph is a path on $t$ vertices.  Likewise, a {\em vertical wire} is a set $\{(x,y+i,1,k): i\in[0,t-1]\}$.  An {\em ell} is the union of a horizontal wire and a vertical wire where there is an edge between one end of the horizontal wire and one end of the vertical wire.  Note that these ends are necessarily in the same unit cell, which we call the {\em corner} of the ell; for an ell $\ell$ we denote the corner by $c(\ell)$.  In the embeddings we study, each chain is an ell.  Our aim now is to specify an orderly way of arranging them into a clique minor.

Looking at Figure \ref{fig:nativeembedding}, one may notice that chains appear in sets that intersect the same unit cells.  With this in mind, for an ell $\ell$ we define its {\em ell block} $(X(\ell),c(\ell))$: $X(\ell)$ is the set of unit cells intersecting $\ell$, and again $c(\ell)$ is the corner of $\ell$, which we must specify for the ell block in order to avoid ambiguity in the case of horizontal or vertical wires of length $1$. We define an {\em ell bundle} $B$ as a (possibly empty) set of vertex-disjoint ells $\ell_1,\dots,\ell_p$ with the same ell blocks, i.e.\ such that $|\{(X(\ell),c(\ell)) \mid \ell \in B\}| \leq 1$.

A \emph{block clique embedding} is a set $\cX$ of $n$ ell blocks $\{(X_1,c_1),\dots,(X_n,c_n)\}$ such that
\begin{itemize}
\item each $X_i$ contains $n$ unit cells (so ells have length $n+1$), and
\item every distinct pair $X_i$, $X_j$ in $\cX$ intersects at exactly one unit cell, which is in the horizontal component of one ell block and the vertical component of the other. 
\end{itemize}

A {\em native clique embedding} respecting a block clique embedding $\cX$ is a collection $\cB$ of ell bundles $\{B_i,\dots,B_n\}$ such that for each $i$ and for each $\ell\in B_i$, $(X_i,c_i) = (X(\ell),c(\ell))$, i.e.\  $(X_i,c_i)$ is the ell block for each ell in $B_i$.  From this definition and the above, we infer that
\begin{itemize}
\item any two ells $\ell$ and $\ell'$ in the same bundle have exactly two edges between them, both in the unit cell $c(\ell) = c(\ell')$, and
\item any two ells $\ell$ and $\ell'$ in different bundles have exactly one edge between them, and it is in the unit cell $X(\ell)\cap X(\ell')$.
\end{itemize}
Hence a native clique embedding is a clique embedding.

It turns out that in a block clique embedding $\cX=\{(X_1,c_1),\dots,(X_n,c_n)\}$, the corners $c_1,\dots,c_n$ form a permutation in the $n\times n$ matrix representing the unit cells of the graph $\mathcal C_{n,n,L}$.  These permutations have a specific structure that is in direct correspondence with a class of permutations representable by circular point sets studied recently by Vatter and Waton \cite{circlepoints}. 

The following theorem provides a constructive classification of block clique embeddings and shows that, in contrast with triangle embeddings, native clique embeddings exist in abundance.

\begin{theorem}\label{thm:blockclique}
  In a $\fc_{n,n,L}$ Chimera graph for $n\geq 2$, there are $4^{n-1}$ block clique embeddings that contain $n$ ell blocks. In particular, they are in natural bijection with the set $\{\mathsf{E,W}\}\times \{\mathsf{NE,NW,SE,SW}\}^{n-2}\times\{\mathsf{N,S}\}$.
\end{theorem}

To prove Theorem \ref{thm:blockclique}, we first show that each ell block in a block clique embedding has a distinct shape. Define the {\em width} of an ell as the number of vertices in its horizontal wire, and its {\em height} as the number of vertices in its vertical wire.  All ells in an ell bundle have the same width and height, so define the width and height of an ell bundle or an ell block as the width and height of its constituent ells. 

\begin{lemma}\label{lem:height}
  Let $\cX=\{(X_1,c_1),\dots,(X_n,c_n)\}$ be a block clique embedding in $\fc_{n,n,L}$. Then the ell blocks of $\cX$ have distinct heights.
\end{lemma}

\begin{proof}
  We will establish that $\cX$ has a unique ell block of height $i$ for each $1 \leq i \leq n$.  Clearly two ells of height 1 or two ells of height $n$ cannot intersect properly (i.e. in exactly one unit which is horizontal for one ell block and vertical for the other).  Assume for contradiction that $\cX$ contains two ells $(X,c)$ and $(X',c')$ of height $1 < i < n$.  Their horizontal and vertical components must occupy different rows and columns respectively. Up to symmetry, there are two cases to consider where ells of the same height intersect properly, shown in Figure \ref{fig:lemmacases}.  In both cases, we name the upper ell block $(X,c)$.
\begin{figure}
  \begin{center}
    \begin{tabular}{ccc}
      \includegraphics{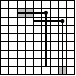}
      &\qquad &
      \includegraphics{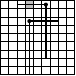}
\vspace{.3em}      \\

{\normalsize      Case 1} & & {\normalsize Case 2}
    \end{tabular}
  \end{center}
\caption{Two cases of intersecting ell blocks with the same shape.\label{fig:lemmacases}}
\end{figure}
  Since there are $n$ ell blocks and $n$ cells per block, every non-corner cell of every ell block must intersect another ell block. In Case 1, we have shaded cells for which the horizontal or vertical coordinate is unique among the two ells.  Consider an ell block $(Y,d)$ which properly intersects $(X,c)$ in a shaded cell. To properly intersect $(X',c')$, $(Y,c)$ cannot intersect $(X,c)$ again, so must intersect $(X',c')$ in a shaded cell also.  Therefore, $Y$ must size greater than $n$, a contradiction. In Case 2, we have shaded the cell which lies directly north of a corner.  In this case, it is clear that no ell block can properly intersect $(X,c)$ at the gray cell and also intersect $(X',c')$ properly. \qed
\end{proof}

Without loss of generality, we will assign labels to each ell block $(X_i,c_i) \in \cX$ so that $X_i$ has height $i$ and width $n-i+1$.

\begin{proof}[of Theorem \ref{thm:blockclique}] We first give a mapping from the set of $4^{n-1}$ words to the set of block clique embeddings and then provide the inverse mapping.

  Let $W=s_1,\dots,s_n$ be a word for which $s_1\in\{\mathsf{E,W}\}$, each of $s_2,\dots,s_{n-1}$ is in $\{\mathsf{NE,NW,SE,SW}\}$, and $s_n\in \{\mathsf{N,S}\}$.  We construct a block clique embedding $\cX=\{(X_1,c_1),\dots,(X_n,c_n)\}$ from $W$ in such a way that each ell block $(X_i,c_i)$ corresponds to the subword $s_i$. We denote the location of a corner $c_i$ by its Cartesian coordinates $(x_i,y_i) \in \{1,\dots,n\}^2$ and coordinates increase to the east (for $x$) and north (for $y$).

  If $s_1=\mathsf{W}$ we place $c_1$ so that it is west of the remaining corners, i.e.\ $x_1=1$.  If $s_1=\mathsf{E}$, we place $c_1$ so that it is east of the remaining corners, i.e.\ $x_1=n$.  We select $y_1$ so that there are $y_1-1$ $\mathsf{S}$s following the subword $s_1$ in the word $W$.  In Figure \ref{fig:thm1first}, we show two initial ell block placements for words where $*$ denotes a (possibly empty) subword containing zero $\mathsf{S}$s.

\begin{figure}
  \begin{center}
    \begin{tabular}{ccc}
      \includegraphics{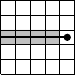}
      &\qquad \qquad &
      \includegraphics{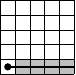}
      \\
      {\normalsize $\mathsf{E*S*S}\,*$} && {\normalsize $\mathsf{W}\,*$}
    \end{tabular}
  \end{center}
\caption{Two examples of the first ell block's orientation and location. $*$ denotes a (possibly empty) subword containing zero $\mathsf{S}$s.   \label{fig:thm1first}}
\end{figure}

  The placement of the first ell block defines a {\em working rectangle} $R_1$ of cells intersecting horizontal wires but not vertical wires in the partially-constructed block clique embedding.  Every time we place an ell block, the working rectangle becomes one unit taller and narrower, and the vertical component of the recently-placed ell covers either the left-most or right-most column of the previous working rectangle.

  Suppose the working rectangle after the selection of $i$ ell blocks is $R_i=[x, x'] \times [y, y']$ in Cartesian coordinates, where $x\leq x'$ and $y \leq y'$.  We choose corner $c_i$ based on subword $s_i$ as follows, noting that $x'=x+n-i-1$ and $y'=y+i-1$:
  \begin{itemize}\itemsep0pt
  \item If $s_i=\mathsf{NE}$ or $\mathsf N$, $c_i = (x',y'+1)$.
  \item If $s_i=\mathsf{NW}$ or $\mathsf N$, $c_i = (x,y'+1)$.
  \item If $s_i=\mathsf{SE}$ or $\mathsf S$, $c_i = (x',y-1)$.
  \item If $s_i=\mathsf{SW}$ or $\mathsf S$, $c_i = (x,y-1)$.
  \end{itemize}
  Letting $c_i = (x_i,y_i)$, we then choose
  $
  X_i = (\{x_i\}\times [y,y']) \cup ( [x,x']\times \{y_i\})
  $
  and update the working rectangle
  $
  R_{i+1} = ([x, x']\backslash \{x_i\}) \times ([y, y'] \cup \{y_i\}).
  $
  Observe that every time that we find a $\mathsf{W}$, $x$ increases by 1 and every time we find a $\mathsf{E}$, $x'$ decreases by 1; every time we find an $\mathsf{N}$, $y$ increases by 1 and every time an $\mathsf{S}$, $y'$ decrease by 1.  Figure \ref{fig:thm1second} shows an example of how the construction might proceed.

\begin{figure}
  \begin{center}
    \begin{tabular}{ccccccccc}
      \includegraphics{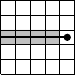} &
      \includegraphics{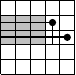} &
      \includegraphics{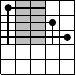} &
      \includegraphics{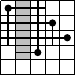} &
      \includegraphics{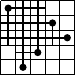}
      \\
      {\normalsize$\mathsf{E}\cdots$} & {\normalsize$\mathsf{ENE}\cdots$} & {\normalsize$\mathsf{ENENW}\cdots$} & {\normalsize$\mathsf{ENENWSE}\cdots$} & {\normalsize$\mathsf{ENENWSES}$}
    \end{tabular}
  \end{center}
\caption{An example of the iterative construction of a block clique embedding.\label{fig:thm1second}}
\end{figure}

  Note that $R_i$ never shares a column with a corner in $\{c_j \mid j\leq i\}$, and the set of rows intersected by $R_i$ is the same set of rows intersected by $\{c_j \mid j\leq i\}$.  It therefore follows from the construction that $\{c_1,\dots,c_n\}$ is a permutation.  Furthermore, by the construction, for $i<j$ there is always a unit cell $X_i\cap X_j$ in the working rectangle $R_i$ -- specifically, it is $(x_j,y_i)$.  So $\cX=\{(X_1,c_1),\dots,(X_n,c_n)\}$ is indeed a block clique embedding.

  Now, we invert our construction to show that it is in fact bijective.  Let $\cX=\{(X_1,c_1),\dots,(X_n,c_n)\}$ be a block clique embedding in which $X_i$ has height $i$; we will reconstruct our word $W$. We call $(x_n,y_1) = X_1 \cap X_n$ the center of $\cX$, and say that a point $(x,y)$ lies to the east or west of the center if $x > x_n$ or $x < x_n$ respectively, and likewise north and south.
  \begin{itemize}\itemsep0pt
  \item Let $s_1 = \mathsf{E}$ if the corner $c_1$ lies to the east of the center, and $\mathsf{W}$ otherwise.
  \item For $1 < i < n$, let $s_i = a_ib_i$, where $a_i = \mathsf{N}$ if $c_i$ lies to the north of the center, and $\mathsf{S}$ otherwise; and $b_i = \mathsf{E}$ if $c_i$ lies to the east of the center, and $\mathsf{W}$ otherwise. 
  \item Let $s_n = \mathsf{N}$ if the corner $c_n$ lies to the north of the center, and $\mathsf{S}$ otherwise.
  \end{itemize}
  Therefore, $W = s_1s_2\ldots s_n$ is in $\{\mathsf{E,W}\}\times \{\mathsf{NE,NW,SE,SW}\}^{n-2}\times\{\mathsf{N,S}\}$.  To see that this construction is in fact the inverse, we observe that each ell block $(X_i,c_i)$ is associated to the direction $s_i$ in which the corner $c_i$ lies from the center, so we obtain the same word we began with. \qed
\end{proof}

\section{Finding optimal native clique embeddings in induced subgraphs\label{sec:3}}

\newcommand{\mpe}{\mathit{maxPartialEmbedding}}
\newcommand{\mb}{\mathit{maxBundle}}
\newcommand{\mext}{\mathit{maxExtension}}
\newcommand{\rin}{\mathit{R_{to}}}
\newcommand{\rout}{\mathit{R_{from}}}
\newcommand{\NCE}{\mathit{NativeCliqueEmbed}}

In this section we describe an algorithm to find a largest native clique embedding $\cB=\{B_i,\dots,B_n\}$ in an induced subgraph $G$ of $\fc_{M,N,L}$ with $n\leq M,N$.  Our algorithm necessarily takes as input a parameter $n$, which determines the size of the chains, i.e.\ $n+1$.

Our algorithm is informed by the proof of Theorem~\ref{thm:blockclique}.  We use dynamic programming to maximize the block clique embeddings with each working rectangle $R$, and do so in an orderly way which results in a polynomial-time algorithm.

As a preprocessing step, for each ell block $(X,c)$ we compute a maximum bundle and store it as $\mb(X,c)$ (it is straightforward to do this in $O(nL)$ time per ell block).  We then use this information to construct block clique embeddings as in the proof of Theorem~\ref{thm:blockclique}: adding one ell block at a time, with working rectangles of increasing height.  Let $(X,c)$ be an ell block of height $i$.  If $1 \leq i \leq n-1$, there is a unique working rectangle $\rout(X,c)$ that can be in effect immediately after $(X,c)$ is placed.  If $2\leq i \leq n$, there is a unique working rectangle $\rin(X,c)$ that can be in effect immediately before $(X,c)$ is placed.  Note that for each working rectangle $R$, the sets
\newcommand{\lbto}{\mathit{X_{to}}}
\newcommand{\lbfrom}{\mathit{X_{from}}}
\begin{equation*}
  \lbfrom(R) :=  \{ (X,c) \mid R = \rin(X,c) \}
\end{equation*}
and
\begin{equation*}
  \lbto(R):=  \{ (X,c) \mid R = \rout(X,c) \}
\end{equation*}
each have size at most four (see Figure \ref{fig:extensions}).

\begin{figure}
  \begin{center}
    \begin{tabular}{ccccc}
      $R$ & $R_1$ & $R_2$ & $R_3$ & $R_4$ \\
      \includegraphics{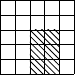} &
      \includegraphics{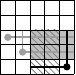} &
      \includegraphics{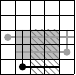} &
      \includegraphics{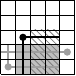} &
      \includegraphics{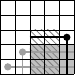}
    \end{tabular}
  \end{center}
  \caption{The four rectangles $R_i = \rin(X_i,c_i)$ with $(X_i,c_i) \in \lbto(R)$.  Gray ell blocks represent precomputed $\mpe(R_i)$ as in Lemma~\ref{lem:mpemax}. }\label{fig:extensions}

\end{figure}
For any a set of ell bundles $\cB = \{B_1, \cdots, B_n\}$ where each $B_i$ is contained in the ell block $(X_i, c_i)$ and  $\cX = \br{(X_1, c_1), \cdots, (X_n, c_n)}$,
\[
\|\cX\| := \abs{\bigcup_{i=1}^n \mb(X_i,c_i)} \geq \abs{\bigcup_{i=1}^n B_i}
\]
provided that the ell blocks are all distinct.  This enables us to construct maximum native clique embeddings while only considering the ell blocks involved.

For each working rectangle $R$ of height $i$ (and width $n-i$), our algorithm computes and stores a \emph{partial block clique embedding} with that working rectangle: a set of ell blocks $\cX_i = \{(X_1,c_1), \cdots, (X_i,c_i)\}$ such that

\begin{itemize}
\item $R = \rout(X_i,c_i)$,
\item $X_j$ has height $j$ for all $1 \leq j \leq i$
\item $\rin(X_j,c_j) = \rout(X_{j-1},c_{j-1})$ for all $2 \leq j \leq i$.
\end{itemize}

In particular, we compute \emph{maximum} partial block clique embeddings with working rectangle $R$; those which maximize $\|\cX_i\|$ over all partial block clique embeddings with a given working rectangle.
We denote a particular maximum partial block clique embedding for a given working rectangle $R$ by $\mpe(R)$, though many may exist.  Our algorithm will operate by extending maximum partial block clique embeddings by ells, so we define $\mext(X,c) = \mpe(\rin(X,c)) \cup \{(X,c)\}$ when $X$ has height $h > 1$, and $\mext(X,c) = \br{(X,c)}$ otherwise.

The following lemma encodes the key step of the algorithm: to find a maximum partial block clique embedding with working rectangle $R$, we need only consider partial block clique embeddings that include an ell block $(X,c)$ such that $R=\rout(X,c)$.

\begin{lemma}\label{lem:mpemax}
  Given a rectangle $R$ of height $i \geq 1$ and width $n-i$,
  \[
  \|\mpe(R)\| = \max_{(X,c) \in \lbto(R)} \abs{\mext(X,c)}.
  \]
\end{lemma}
\begin{proof} We proceed by induction on $i$.  When $i=1$, our claim follows immediately from definitions.  Assume that $i>1$ and that our claim holds for all working rectangles of height $i-1$.

  We consider the complete set of partial block clique embeddings with working rectangle $R$,
  $
  S = \{\cX_i \mid \rout(X_i,c_i) = R\} = \{\cX_i \mid (X_i,c_i) \in \lbto(R)\}
  $
  where $(X_i,c_i)$ is the ell block with height $i$ in $\cX_i$.  By definition,
  $
  \|\mpe(R)\| = \max_{\cX_i \in S} \|\cX_i\|.
  $
  For contradiction, pick some maximum $\cX_i \in S$ and suppose that
  \[
  \|\cX_i\| > \|\mpe(\rin(X,c)) \cup \{(X,c)\}\|
  \]
  for all $(X,c) \in \lbto(R)$.  In particular, letting $\cX_{i-1} = \cX_i \setminus \br{(X_i,c_i)}$,
  \begin{eqnarray*}
    \|\cX_i\|&=& \norm{\cX_{i-1}} + \norm{\br{(X_i,c_i)}}\\ &>& \|\mpe(\rin(X_i,c_i))\| + |\mb(X_i,c_i)|,
  \end{eqnarray*}
  a contradiction since $\cX_{i-1}$ has working rectangle $\rin(X_i,c_i)$. \qed
\end{proof}

We now present the algorithm.  The idea follows from Lemma~\ref{lem:mpemax}, to compute $\mpe(R)$ for rectangles of increasing height.  To do so we treat the set of possible working rectangles as a digraph, where $R\rightarrow R'$ if and only if there is an ell block $(X,c)$ for which $R\in \rout(X,c)$ and $R'\in\rin(X,c)$.  This means that if $R\rightarrow R'$, the height of $R'$ is one more than the height of $R$.  The number of edges in this digraph is equal to the number of ell blocks.  To compute $\mpe(R')$, assuming that we have computed $\mpe(R)$ for all rectangles of lesser height, we simply set $\mpe(R')$ to be a maximum partial block clique embedding in the set
$$\left\{ \mext(X,c) \mid (X,c)\in \lbto(R') \right\}.$$
Once we have computed $\mpe(R)$ for all rectangles $R$ of height $n-1$ and width $1$, we pick a maximum-sized clique embedding from the set
$$\left\{ \mext(X,c) \mid (X,c)\textrm{ has height }n-1\textrm{ and width }1 \right\}.$$
Pseudocode is given in Algorithm \ref{alg:cliqueembed}.
\begin{algorithm}
  \begin{algorithmic}[1]

    \Function{NativeCliqueEmbed}{$G$, $n$}
    \For{$i=1,\dots,n-1$}
    \For{each rectangle $R$ of height $i$ and width $n-i$}
    \State $\mpe(R)\gets \emptyset$
    \EndFor
    \For{each ell block $(X,c)$ of height $i$ and width $n-i+1$\label{line:forblock}}
    \State $\cB \gets\mpe(\rin(X,c))\cup\{(X,c)\}$
    \If{$\|\mpe(\rout(X,c))\| < \|\cB\|$}
    \State $\mpe(\rout(X,c)) \gets \cB$
    \EndIf
    \EndFor
    \EndFor
    \State $\cB_{\mathit{max}} \gets \emptyset$ \label{line:finalpass}
    \For{each ell block $(X,c)$ of height $n$ and width $1$}
    \State $\cB \gets \mpe(\rin(X,c)) \cup\br{(X,c)}$ \label{line:lastextension}
    \If{$\|\cB_{\mathit{max}}\| < \|\cB\|$}
    \State $\cB_{\mathit{max}} \gets \cB$
    \EndIf
    \EndFor
    \State \Return{$\br{\mb(X,c,G) \mid (X,c) \in \cB_{\mathit{max}}}$} \label{line:return}
    \EndFunction
  \end{algorithmic}
  \caption{The algorithm to find a maximum-sized native clique embedding in an induced subgraph of a Chimera graph.}\label{alg:cliqueembed}
\end{algorithm}

\begin{theorem} 
  The $\NCE$ algorithm finds a maximum-sized native clique embedding with chain length $n+1$ in polynomial time.
\end{theorem}
\begin{proof}
  We first prove correctness, then the bound on running time.

  Note that the loop beginning on line~\ref{line:forblock} of Algorithm~\ref{alg:cliqueembed} iterates over all ell blocks of height $i$ and width $n-i$.  Given a rectangle $R$, there are up to four ell blocks $(X,c)$ for which $R = \lbto(X,c)$.  If we ignore all ell blocks except those incident to a particular rectangle $R$, this loop implements Lemma~\ref{lem:mpemax} directly. Therefore, when we reach line~\ref{line:finalpass}, we have computed $\mpe(R)$ for all $R$ with height $n-1$ and width 1.

  Let $\cB = \br{B_1,\cdots,B_n}$ be a maxmimum native clique embedding where $B_i$ is an ell bundle in the ell block $(X_i,c_i)$ with height $i$.  By Theorem~\ref{thm:blockclique}, $\cX_{n-1} = 
  \br{(X_1,c_1),\cdots,(X_{n-1},c_{n-1})}$ is a partial block clique embedding with working rectangle $\rin(X_n,c_n)$.  By Lemma~\ref{lem:mpemax}, \[\norm{\mpe(\rin(X_n,c_n))} \geq \norm{\cX_{n-1}},\] so we see a clique embedding of size at least $\norm{\cX_{n-1}} + |\mb(X_n,c_n)|$ when line~\ref{line:lastextension} is reached with $(X,c) = (X_n,c_n)$.  Therefore, when line~\ref{line:return} is reached, $\norm{\cB_{max}} \geq \norm{\cB}$ and a native clique embedding of maximum size has been found.

  We can compute $\mb(X,c,G)$ in $O(nL)$ time, and there are polynomially many ell blocks and rectangles:  There are at most $MN$ possible locations of a rectangle's lower-left corner, and $n$ possible shapes, thus at most $nMN$ rectangles.  Likewise there are at most $MN$ possible locations for an ell block's corner, and at most $4(n-1)$ ell blocks containing $n$ unit cells with a given corner, thus at most $4nMN$ ell blocks.

  It follows that each line in Algorithm \ref{alg:cliqueembed} is evaluated $O(nMN)$ times, and the preprocessing step of computing $\mb(X,c,G)$ for each ell block $(X,c)$ naively takes $O(n^2MNL)$ time.  Consequently, with the rough bound that each line in Algorithm \ref{alg:cliqueembed} takes $O(nL)$ time for a single evaluation, we can bound the total running time of our algorithm by $O(n^2MNL)$. \qed
\end{proof}

\paragraph{Remark.} The $O(n^2MNL)$-time bound on Algorithm \ref{alg:cliqueembed} is quadratic in the number of vertices in $\fc_{N,M,L}$, i.e.\ $O(n^2MNL)\subseteq O((MNL)^2)$. Assume $M  \leq N$ and $L$ is constant. With a little more care, we can modify Algorithm \ref{alg:cliqueembed} to achieve a bound of $O(N^3)$ instead of $O(N^4)$. Doing this involves (a) precomputing all maximum horizontal and vertical line bundles in time $O(LNM^2 + LMN^2)$ with a dynamic programming approach, which allows us to compute $|\mb(X,c,G)|$ in $O(1)$ time, and (b) exploiting the fact that throughout the algorithm, we need only keep track of the size of maximum partial embeddings and the route used to reach it (replacing $\mpe(R)$ with a mapping $R \mapsto (X,c) \in \lbto(R)$), rather than the embeddings themselves.

NativeCliqueEmbed gives a maximum native clique embedding for a fixed chain length.  To find a maximum native clique embedding over all chain lengths for a given graph, we simply repeat the process for each choice of $n\in \{2,\dots,\min\{M,N\}\}$. For $M  \leq N$ and constant $L$, this gives an overall running time on a subgraph of $\fc_{M,N,L}$ of $O(N^5)$ with the naive implementation and $O(N^4)$ with the refinement discussed above.

\subsection{Induced and general subgraphs}

We now discuss the motivation of using induced subgraphs and how to approach more general subgraphs.

Recall that we have restricted our attention to induced subgraphs rather than more general subgraphs because failed couplers adjoining working qubits are relatively rare.   In an induced subgraph, we will still focus on horizontal and vertical wires, and it is easy to find a maximum set of such wires in a line of unit cells.

A {\em maximum ell bundle} is a maximum-sized set of vertex-disjoint ells that occupy an ell block, and the {\em size} of an ell bundle is the number of ells contained therein.  It is simple to find a maximum ell bundle in a given ell block: finding maximum sets $S_H$ and $S_V$ of horizontal and vertical wires spanning a line of cells is trivial, and these lines may be paired off arbitrarily since the corner is a complete bipartite graph.  Here there may be difficulty in generalizing even this relatively simple optimization problem in the face of arbitrary edge deletion.  Finding a maximum ell bundle in an arbitrary Chimera subgraph is polynomially equivalent to finding a maximum clique in $D_2(\mathcal{L}(B))$ where $B$ is a bipartite graph, $\mathcal{L}(B)$ is the line graph of $B$, and $D_2(G)$ is the distance-2 graph with vertex set $V(G)$ and edge set $\{uv : d_G(u,v) = 2\}$.  We are unsure of the complexity of this problem, but expect that it is NP-complete.

Note that we can easily relax the requirement of an induced subgraph when couplers between unit cells are defective -- in any case, one just computes the number of wires in a line of cells where all qubits and couplers are contained in the subgraph.  In short, failed inter-cell couplers don't increase the difficulty of the problem.

If we restrict our attention to Chimera graphs $\fc_{N,N,L}$ where $L = O(\log N)$, or assume $L$ to be a constant, then this difficulty at the corners can be swept under the rug.  At present, this appears to be a reasonable consideration, as it is much easier from a manufacturing and design standpoint to increase $N$ than it is to increase $L$.

However, there is a further challenge introduced by intra-cell couplers.  Our algorithm intrinsically relies upon the assumption that couplers exist between any two ells whose unit cells intersect.  Missing intra-cell couplers void that assumption.  The easiest remedy for this obstruction is to consider vertex covers of the failed intra-cell couplers: Given a graph $G$ that is the subgraph of $ \fc_{M,N,L}$ induced by vertex set $W$, i.e.\ $G = \fc_{M,N,L}[W]$, let $U_1, \dots, U_s$ be the list of minimal vertex covers of the failed intra-cell couplers.  Then, for each $1 \leq i \leq s$, compute a largest clique minor considering $G[W]$ for the purpose of constructing ells, and $G[W \setminus U_i]$ for the purpose of growing cliques.  This approach gives a fixed-parameter tractable algorithm for finding the maximum native clique embedding in an arbitrary subgraph of Chimera in terms of the number of missing edges with both endpoints intact -- for the D-Wave Two processors installed at NASA Ames \cite{Venturelli2014} and ISI \cite{Albash2015}, this parameter was zero.

\section{Comparison with previous work\label{sec:previouswork}}
Klymko, Sullivan and Humble gave a greedy embedding algorithm that quickly produces similar embeddings to those in this paper~\cite{Klymko2014}.  Their algorithm produces plus-shaped chains (having nearly twice as many qubits), and restricts its search to embeddings where the vertical and horizontal components of pluses are only allowed to meet at the diagonal of a fixed square.  While our algorithm is slower, taking $O(N^4)$ time compared to their $O(N^3)$ for $\fc_{N,N,4}$, it is exhaustive, empirically embeds larger cliques, and produces chains of roughly half the size.

\begin{figure}
  \begin{center}
    \includegraphics{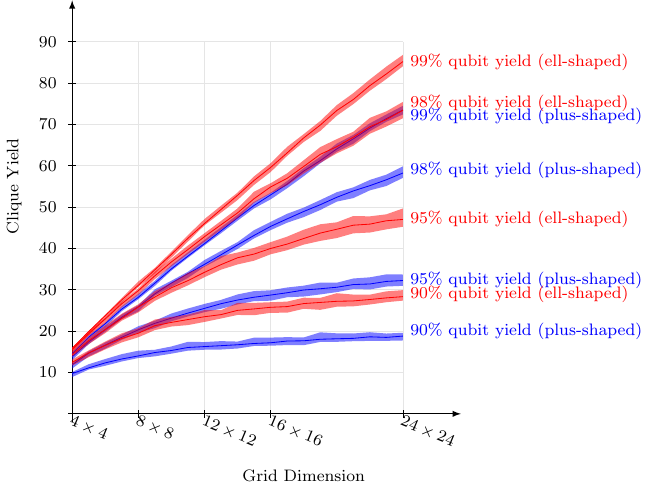}
  \end{center}
  \caption{(Color online) Comparison of clique yields for ell-shaped (red) and plus-shaped (blue) chains.  The solid lines denote the median, and the shaded regions encompass the middle two quartiles.}\label{fig:ellvplus}
\end{figure}

Given a family of clique minors, the {\em clique yield} of a graph $G$ over that family is the size of the largest clique minor in the family, in $G$.  In Figure~\ref{fig:ellvplus}, we compare ell- and plus-shaped chains as the grid size grows for several fixed percentages of operational qubits.  In both families, a similar asymptotic behavior becomes clear: for a fixed qubit failure rate, increasing the grid size gives diminishing returns in terms of clique yield.  However, the difference between these curves is significant, with ell-shaped chains producing much larger clique minors.

\bibliography{bibtex}

\begin{thebibliography}{10}
\providecommand{\url}[1]{{#1}}
\providecommand{\urlprefix}{URL }
\expandafter\ifx\csname urlstyle\endcsname\relax
  \providecommand{\doi}[1]{DOI~\discretionary{}{}{}#1}\else
  \providecommand{\doi}{DOI~\discretionary{}{}{}\begingroup
  \urlstyle{rm}\Url}\fi

\bibitem{Albash2015}
Albash, T., Vinci, W., Mishra, A., Warburton, P.A., Lidar, D.A.: Consistency
  tests of classical and quantum models for a quantum annealer.
\newblock Physical Review A \textbf{91}(4), 042,314 (2015)

\bibitem{Boixo2014}
Boixo, S., Smelyanskiy, V.N., Shabani, A., Isakov, S.V., Dykman, M., Denchev,
  V.S., Amin, M., Smirnov, A., Mohseni, M., Neven, H.: Computational role of
  collective tunneling in a quantum annealer.
\newblock arXiv preprint arXiv:1411.4036  (2014)

\bibitem{Cai2014}
Cai, J., Macready, W., Roy, A.: A practical heuristic for finding graph minors.
\newblock arXiv preprint arXiv:1406.2741  (2014)

\bibitem{Choi2008}
Choi, V.: {Minor-embedding in adiabatic quantum computation: I. The parameter
  setting problem}.
\newblock Quantum Information Processing \textbf{7}(5), 193--209 (2008)

\bibitem{Choi2011}
Choi, V.: {Minor-embedding in adiabatic quantum computation: II.
  Minor-universal graph design}.
\newblock Quantum Information Processing \textbf{10}(3), 343--353 (2011)

\bibitem{Dickson2013}
Dickson, N., et~al.: {Thermally assisted quantum annealing of a 16-qubit
  problem}.
\newblock Nature Communications \textbf{4}(May), 1903 (2013).
\newblock \doi{10.1038/ncomms2920}.
\newblock \urlprefix\url{http://www.ncbi.nlm.nih.gov/pubmed/23695697}

\bibitem{Dziarmaga2005}
Dziarmaga, J.: Dynamics of a quantum phase transition: Exact solution of the
  quantum ising model.
\newblock Physical review letters \textbf{95}(24), 245,701 (2005)

\bibitem{Johnson2011}
Johnson, M., Amin, M., Gildert, S., Lanting, T., Hamze, F., Dickson, N.,
  Harris, R., Berkley, A., Johansson, J., Bunyk, P., et~al.: Quantum annealing
  with manufactured spins.
\newblock Nature \textbf{473}(7346), 194--198 (2011)

\bibitem{King2014}
King, A.D., McGeoch, C.C.: Algorithm engineering for a quantum annealing
  platform.
\newblock arXiv preprint arXiv:1410.2628  (2014)

\bibitem{Klymko2014}
Klymko, C., Sullivan, B.D., Humble, T.S.: Adiabatic quantum programming: minor
  embedding with hard faults.
\newblock Quantum Information Processing \textbf{13}(3), 709--729 (2014)

\bibitem{Perdomo-Ortiz2015}
Perdomo-Ortiz, A., Fluegemann, J., Biswas, R., Smelyanskiy, V.N.: A performance
  estimator for quantum annealers: {Gauge} selection and parameter setting.
\newblock arXiv preprint arXiv:1503.01083  (2015)

\bibitem{circlepoints}
Vatter, V., Waton, S.: On points drawn from a circle.
\newblock Electronic Journal of Combinatorics \textbf{18}(1), P223 (2011)

\bibitem{Venturelli2014}
Venturelli, D., Mandr{\`a}, S., Knysh, S., O'Gorman, B., Biswas, R.,
  Smelyanskiy, V.: Quantum optimization of fully-connected spin glasses.
\newblock arXiv preprint arXiv:1406.7553  (2014)

\bibitem{Young2013}
Young, K., Blume-Kohout, R., Lidar, D.: {Adiabatic quantum optimization with
  the wrong Hamiltonian}.
\newblock Physical Review A \textbf{88}(6), 062,314 (2013)

\end{thebibliography}

\end{document}